\documentclass[letterpaper, 12pt]{article}
\usepackage[top=1in,left=1in,right=1in,bottom=1in]{geometry}
\usepackage{mathtools}
\usepackage{amsmath}
\usepackage{graphicx}
\usepackage[]{algorithm2e}
\usepackage{ulem}
\usepackage{amssymb}
\usepackage{amsthm}
\theoremstyle{definition}

\newtheorem{theorem}{Theorem}
\newtheorem{corollary}{Corollary}
\newtheorem{lemma}{Lemma}

\newtheorem{remark}{Remark}
\usepackage{hhline}
\usepackage[sectionbib]{natbib}
\usepackage{authblk}
\usepackage{amsfonts}
\usepackage{multirow}

\renewcommand{\theequation}{\thesection.\arabic{equation}}

\begin{document}

\fontsize{12}{14pt plus.8pt minus .6pt}\selectfont \vspace{0.8pc}
\centerline{\large\bf Variable screening with multiple studies}
\vspace{.4cm} \centerline{Tianzhou Ma$^1$, Zhao Ren$^2$ and George C. Tseng$^1$} \vspace{.4cm} \centerline{\it 
$^1$Department of Biostatistics, University of Pittsburgh} \par \centerline{\it $^2$Department of Statistics, University of Pittsburgh} \vspace{.55cm} \fontsize{9}{11.5pt plus.8pt minus
.6pt}\selectfont

\begin{abstract}
Advancement in technology has generated abundant high-dimensional data that allows integration of multiple relevant studies. 
Due to their huge computational advantage, variable screening methods based on marginal correlation have become promising alternatives to the popular regularization methods for variable selection. However, all these screening methods are limited to single study so far. In this paper, we consider a general framework for variable screening with multiple related studies, and further propose a novel two-step screening procedure using a self-normalized estimator for high-dimensional regression analysis in this framework. Compared to the one-step procedure and rank-based sure independence screening (SIS) procedure, our procedure greatly reduces false negative errors while keeping a low false positive rate. Theoretically, we show that our procedure possesses the sure screening property with weaker assumptions on signal strengths and allows the number of features to grow at an exponential rate of the sample size. In addition, we relax the commonly used normality assumption and allow sub-Gaussian distributions. Simulations and a real transcriptomic application illustrate the advantage of our method as compared to the rank-based SIS method. 

\end{abstract}

\fontsize{12}{14pt plus.8pt minus .6pt}\selectfont

\setcounter{equation}{0} 
\section{Introduction}
\label{sec:intro}

In many areas of scientific disciplines nowadays such as omics studies (including genomics, transcriptomics, etc.), biomedical imaging and signal processing, high dimensional data with much greater number of features than the sample size (i.e. $p>>n$) have become rule rather than exception. For example, biologists may be interested in predicting certain clinical outcome (e.g. survival) using the gene expression data where we have far more genes than the number of samples. With the advancement of technologies and affordable prices in recent biomedical research, more and more experiments have been performed on a related hypothesis or to explore the same scientific question. Since the data from one study often have small sample size with limited statistical power, effective information integration of multiple studies can improve statistical power, estimation accuracy and reproducibility. Direct merging of the data (a.k.a. ``mega-analysis") is usually less favored due to the inherent discrepancy among the studies \citep{tseng2012comprehensive}. New statistical methodologies and theories are required to solve issues in high-dimensional problem when integrating multiple related studies. 

Various regularization methods have been developed in the past two decades and frequently used for feature selection in high-dimensional regression problems. Popular methods include, but are not limited to, Lasso \citep{tibshirani1996regression}, SCAD \citep{fan2001variable}, elastic net \citep{zou2005regularization} and adaptive Lasso \citep{zou2006adaptive}. When group structure exists among the variables (for example, a set of gene features belonging to a pre-specified pathway), group version of regularization methods can be applied \citep{yuan2006model, meier2008group, nardi2008asymptotic}. One can refer to \cite{fan2010selective} and \cite{huang2012selective} for a detailed overview of variable selection and group selection in high-dimensional models. When the number of features grows significantly larger than the sample size, most regularization methods perform poorly due to the simultaneous challenges of computation expediency, statistical accuracy and algorithmic stability \citep{fan2009ultrahigh}. Variable screening methods become a natural way to consider by first reducing to a lower or moderate dimensional problem and then performing variable regularization.  \cite{fan2008sure} first proposed a sure independent screening (SIS) method to select features based on their marginal correlations with the response in the context of linear regression models and showed such fast selection procedure enjoyed a ``sure screening property". Since the development of SIS, many screening methods have been proposed for generalized linear models \citep{fan2009ultrahigh,fan2010sure,chang2013marginal}, nonparametric additive models or semiparametric models \citep{fan2011nonparametric,chang2016local}, quantile linear regression \citep{ma2017variable}, Gaussian graphical models \citep{luo2014sure,liang2015equivalent} or exploit more robust measures for sure screening \citep{zhu2011model,li2012feature,li2017variable}. However, all these screening methods are limited to single study so far. 

In this paper, we first propose a general framework for simultaneous variable screening with multiple related studies. Compared to single study scenario, inclusion of multiple studies gives us more evidence to reduce dimension and thus increases the accuracy and efficiency of removing unimportant features during screening. To our knowledge, our paper is the first to utilize multiple studies to help variable screening in high-dimensional linear regression model. Such a framework provides a novel perspective to the screening problem and opens a door to the development of methods using multiple studies to perform screening under different types of models or with different marginal utilities. In this framework, it is natural to apply a selected screening procedure to each individual study, respectively. However, important features with weak signals in some studies may be incorrectly screened out if only such a one-step screening is performed. To avoid such false negative errors and fully take advantage of multiple studies, we further propose a two-step screening procedure, where one additional step of combining studies with potential zero correlation is added to the one-step procedure for a second check. This procedure has the potential to save those important features with weak signals in individual studies but strong aggregate effect across studies during the screening stage. Compared to the naive multiple study extension of SIS method, our procedure greatly reduces the false negative errors while keeping a low false positive rate. These merits are confirmed by our theoretical analysis. Specifically, we show that our procedure possesses the sure screening property with weaker assumptions on signals and allows the number of features to grow at an exponential rate of the sample size. Furthermore, we only require the data to have sub-Gaussian distribution via using novel self-normalized statistics. Thus our procedure can be applied to more general distribution family other than Gaussian distribution, which is considered in \cite{fan2008sure} and \cite{buhlmann2010variable} for a related screening procedure under single study scenarios. After screening, we further apply two general and applicable variable selection algorithms: the multiple study extension of PC-simple algorithm proposed by \cite{buhlmann2010variable} as well as a two-stage feature selection method to choose the final model in a lower dimension. 

The rest of the paper is organized as follows. In Section 2, we present a framework for variable screening with multiple related studies as well as notations. Then we propose our two-step screening procedure in Section 3. Section 4 provides the theoretical properties of our procedure, and demonstrates the benefits of multiple related studies as well as the advantages of our procedure. General algorithms for variable selection that can follow from our screening procedure are discussed in Section 5. Section 6 and 7 include the simulation studies and a real data application on three breast cancer transcriptomic studies, which illustrate the advantage of our method in reducing false negative errors and retaining important features as compared to the rank-based SIS method. We conclude and discuss possible extensions of our procedure in Section 8. Section 9 provides technical proofs to the major theorems.  

\setcounter{equation}{0} 
\section{Model and Notation}
\label{sec:model}

Suppose we have data from $K$ related studies, each has $n$ observations. Consider a random design linear model in each study $k \in [K]$ ($[K]=1,\ldots, K$): 
 \begin{equation}
   \label{eq1}
  Y^{(k)}  =  \sum_{j=1}^p  \beta^{(k)}_j  X^{(k)}_j  + \epsilon^{(k)},   
\end{equation} 
where each $Y^{(k)}\in \mathbb{R}$, each $X^{(k)} = (X^{(k)}_1,\ldots, X^{(k)}_p)^T \in \mathbb{R}^{p}$ with $E(X^{(k)})=\mu_X^{(k)}$ and cov$(X^{(k)})=\Sigma_X^{(k)}$, each $\epsilon^{(k)} \in \mathbb{R}$ with $E(\epsilon^{(k)})=0$ and var$(\epsilon^{(k)})=\sigma^2$ such that $\epsilon^{(k)}$ is uncorrelated with $X^{(k)}_1,\ldots, X^{(k)}_p$, and $\beta^{(k)} = (\beta^{(k)}_1, \ldots, \beta^{(k)}_p)^T \in \mathbb{R}^{p}$. We assume implicitly with $E(Y^{(k)2})<\infty$ and $E\{(X^{(k)}_j)^2\}<\infty$ for $j\in [p]$ ($[p]=1,\ldots, p$). 

When $p$ is very large, we usually assume that only a small set of covariates are true predictors that contribute to the response. In other words, we assume most of $\beta_j=(\beta^{(1)}_j, \ldots, \beta^{(K)}_j)^T$, where $j\in[p]$, are equal to a zero vector. In addition, in this paper, we assume $\beta_j^{(k)}$'s are either zero or non-zero in all $K$ studies. This framework is partially motivated by a high-dimensional linear random effect model considered in literature (e.g.,\cite{jiang2016high}). More specifically, we can have $\beta=(\beta_{(1)}^T, 0^T)^T$, where $\beta_{(1)}$ is the vector of the first $s_0$ non-zero components of $\beta$ ($1\le s_0\le p$). Consider a random effect model where only the true predictors of each study are treated as the random effect, that is, $\beta^{(k)}=
(\beta_{(1)}^{(k)},0)^T$ and $\beta_{(1)}^{(k)}$ is distributed as $N(\beta_{(1)},\tau^2I_{s_0})$, where $\tau^2$ is independent of $\epsilon$ and $X$. Consequently, $\beta_j^{(k)}$'s are either zero or non-zero in all $K$ studies with probability one. Such assumption fits the reality well, for example, in a typical GWAS study, a very small pool of SNPs are reported to be associated with a complex trait or disease among millions \citep{jiang2016high}. 

With $n$ i.i.d. observations from model (\ref{eq1}), our purpose is to identify the non-zero $\beta_{(1)}$, thus we define the following index sets for active and inactive predictors:
\begin{equation}
\label{eq2}
\begin{split}
\mathcal{A} & = \{j\in [p]; \beta_j  \neq 0\} =\{j\in [p];  \beta_j^{(k)} \neq 0 \text{ for all } k \}; \\  
\mathcal{A}^C & = \{j\in [p]; \beta_j = 0\} =   \{j\in [p];  \beta_j^{(k)} = 0 \text{ for all }  k \}, 
\end{split}
\end{equation}
where $\mathcal{A}$ is our target. Clearly, under our setting, $\mathcal{A}$ and $\mathcal{A}^C$ are complementary to each other so that the identification of $\mathcal{A}^C$ is equivalent to the identification of $\mathcal{A}$. Let $|\mathcal{A}|=s_0$, where $|\cdot|$ denotes the cardinality.

\setcounter{equation}{0} 
\section{Screening procedure with multiple studies} 
\label{sec:procedure}

\subsection{Sure independence screening}
\label{subsec:sis}

For a single study ($K=1$), \cite{fan2008sure} first proposed the variable screening method called sure independence screening (SIS) which ranked the importance of variables according to their marginal correlation with the response and showed its great power in preliminary screening and dimension reduction for high-dimensional regression problems. \cite{buhlmann2010variable} later introduced the partial faithfulness condition that a zero partial correlation for some separating set $S$ implied a zero regression coefficient and showed that it held almost surely for joint normal distribution. In the extreme case when $S=\emptyset$, it is equivalent to the SIS method. 

The purpose of sure screening is to identify a set of moderate size $d$ (with $d<<p$) that will still contain the true set $\mathcal{A}$. Equivalently, we can try to identify $\mathcal{A}^C$ or subsets of $\mathcal{A}^C$ which contain unimportant features that need to be screened out. There are two potential errors that may occur in any sure screening methods \citep{fan2010selective}: 
\begin{enumerate}
\item \textbf{False Negative (FN):} Important predictors that are marginally uncorrelated but jointly correlated with the response fail to be selected.
\item \textbf{False Positive (FP):} Unimportant predictors that are highly correlated with the important predictors can have higher priority to be selected than other relatively weaker important predictors.
\end{enumerate}

The current framework for variable screening with multiple studies is able to relieve us from the FP errors significantly. Indeed, we have multiple studies in our model setting thus we have more evidence to exclude noises and reduce FP errors than single study. In addition, sure screening is used to reduce dimension at a first stage, so we can always include a second stage variable selection methods such as Lasso or Dantzig selection to further refine the set and reduce FP errors. 

The FN errors occur when signals are falsely excluded after screening. Suppose $\rho_j$ is the marginal correlation of the $j$th feature with the response, with which we try to find the set $\{j: \rho_j=0 \}$ to screen out. Under the assumption of partial faithfulness (for explicit definition, see Section \ref{subsec:sure}), these variables have zero coefficients for sure so the FN errors are guaranteed to be excluded. However, this might not be true for the empirical version of marginal correlation. For a single study ($K=1$), to rule out the FN errors in empirical case, it is well-known that the signal-to-noise ratio has to be large (at least of an order of $(\log p /n)^{1/2}$ after Bonferroni adjustment). In the current setting with multiple studies, the requirement on strong signals remains the same if we naively perform one-step screening in each individual study. As we will see next, we propose a novel two-step screening procedure which allows weak signals in individual studies as long as the aggregate effect is strong enough. Therefore our procedure is able to reduces FN errors in the framework with multiple studies.  

Before closing this section, it is worthwhile to mention that, to perform a screening test, one usually applies Fisher's z-transformation on the sample correlation \citep{buhlmann2010variable}. However, this will require the bivariate normality assumption. Alternatively, in this paper, we propose to use the self-normalized estimator of correlation that works generally well even for non-Gaussian data \citep{shao1999cramer}. Similar ideas have been applied in the estimation of large covariance matrix \citep{cai2016large}. 

\subsection{Two-step screening procedure with multiple studies} 
\label{subsec:twostep}

In the presence of multiple studies, we have more evidence to reduce dimension and $\rho_j^{(k)}=0$ for any $k$ will imply a zero coefficient for that feature. On one hand, it is possible for features with zero $\beta_j$ to have multiple non-zero $\rho_j^{(k)}$'s. On the other hand, a non-zero $\beta_j$ will have non-zero $\rho_j^{(k)}$'s in all studies. Thus, we aim to identify the following two complementary sets while performing screening with multiple studies:

\begin{equation}
    \label{eq3}
   \begin{split}
\mathcal{A}^{[0]}  &=  \{j\in [p]; \quad \min_k |\rho_j^{(k)}| = 0 \}, \\ 
\mathcal{A}^{[1]}  &=  \{j\in [p]; \quad \min_k |\rho_j^{(k)}| \neq 0 \}.
  \end{split}
\end{equation}

We know for sure that $ \mathcal{A}^{[0]}  \subseteq \mathcal{A}^C$ and $ \mathcal{A}  \subseteq \mathcal{A}^{[1]}$ with the partial faithfulness assumption. For $j \in \mathcal{A}^{[0]}$, the chance of detecting a zero marginal correlation in at least one study has been greatly increased with increasing $K$, thus unimportant features will more likely be screened out as compared to single study scenario.   

One way to estimate $\mathcal{A}^{[1]}$ is to test $H_0: \rho_j^{(k)}=0$ of each $k$ for each feature $j$. When any of the $K$ tests is not rejected for a feature, we will exclude this feature from $\mathcal{\hat{A}}^{[1]}$ (we call it the ``One-Step Sure Independence Screening" procedure, or ``OneStep-SIS" for short). This can be viewed as an extension of the screening test to multiple study scenario. However, in reality, it is possible for important features to have weak signals thus small $|\rho_j^{(k)}|$'s in at least one study. These features might be incorrectly classified into $\mathcal{\hat{A}}^{[0]}$ since weak signals can be indistinguishable from null signals in individual testing. It will lead to the serious problem of false exclusion of important features (FN) from the final set during screening. 

This can be significantly improved by adding a second step to combine those studies with potential zero correlation (i.e., fail to reject the null $H_0: \rho_j^{(k)}=0$) identified in the first step and perform another aggregate test. For the features with weak signals in multiple studies, as long as their aggregate test statistics is large enough, they will be retained. Such procedure will be more conservative in screening features as to the first step alone, but will guarantee to reduce false negative errors.  

For simplicity, we assume $n$ i.i.d. observations $(X^{(k)}_i,Y_i^{(k)})$, $i\in[n]$, are obtained from all $K$ studies. It is straightforward to extend the current procedure and analysis to the scenarios with different sample sizes across multiple studies, and thus omitted.  Our proposed ``Two-Step Aggregation Sure Independence Screening" procedure (``TSA-SIS" for short) is formally described below: \\

\noindent{\bf Step 1. Screening in each study} 

In the first step, we perform screening test in each study $k\in [K]$ and obtain the estimate of study set with potential zero correlations $\hat{l}_{j}$  for each $j\in [p]$ as:
\begin{equation}
  \label{eq4}
   \hat{l}_{j}  = \{k;  |\hat{T}^{(k)}_{j}| \le \Phi^{-1} (1- \alpha_{1}/2) \} \quad \text{ and } \quad  \hat{T}^{(k)}_{j} = \frac{\sqrt{n} \hat{\sigma}_j^{(k)} }{\sqrt{ \hat{\theta}_j^{(k)}  } }, 
\end{equation}
where $\hat{\sigma}_j^{(k)} = \frac{1}{n} \sum_{i=1}^n (X_{ij}^{(k)} - \bar{X}_j^{(k)} ) (Y_i^{(k)} - \bar{Y}^{(k)}) $ is the sample covariance and $\hat{\theta}_j^{(k)} =\frac{1}{n} \sum_{i=1}^n [ (X_{ij}^{(k)} - \bar{X}_j^{(k)} ) (Y_i^{(k)} - \bar{Y}^{(k)}) - \hat{\sigma}_j^{(k)} ]^2 $. $\hat{T}^{(k)}_{j}$ is the self-normalized estimator of covariance between $X^{(k)}_j$ and $Y^{(k)}$. $\Phi$ is the CDF of standard normal distribution and $\alpha_{1}$ the pre-specified significance level. 

In each study, we test if $|\hat{T}^{(k)}_{j}| > \Phi^{-1} (1- \alpha_{1}/2)$, if not, we will include study $k$ into $\hat{l}_{j}$. This step does not screen out any variables, but instead separates potential zero and non-zero study-specific correlations for preparation of the next step. Define the cardinality of $\hat{l}_{j}$ as $\hat{\kappa}_{j} = | \hat{l}_{j}|$. If $\hat{\kappa}_{j}=0$ (i.e., no potential zero correlation), we will for sure retain feature $j$ and not consider it in step 2; Otherwise, we move on to step 2. 

\begin{remark} \label{rem:scaling}
By the scaling property of $\hat{T}^{(k)}_{j}$, it is sufficient to impose assumptions on the standardized variables: $W^{(k)} = \frac{Y^{(k)} - E(Y^{(k)} ) }{\sqrt{\text{var}(Y^{(k)} )}},  Z^{(k)}_j = \frac{X^{(k)}_j - E(X^{(k)}_j) }{\sqrt{\text{var}(X^{(k)}_j )}}$. Thus $\hat{T}^{(k)}_{j}$ can also be treated as the self-normalized estimator of correlation. We thus can define $\theta_j^{(k)} =\text{var}(Z^{(k)}_j W^{(k)})$ and $\sigma_j^{(k)}= \text{cov}(Z^{(k)}_j , W^{(k)}) = \rho_j^{(k)}$. 
\end{remark}

\begin{remark}
In our analysis, the index set in (\ref{eq4}) is shown to coincide with $l_j(j\in \mathcal{A}^{[0]} )$ and $l_j(j\in \mathcal{A}^{[1]} )$ which will be introduced in more details in Section \ref{sec:theory}. \\
\end{remark}

\noindent{\bf Step 2. Aggregate screening} 

In the second step, we wish to test whether the aggregate effect of potential zero correlations in $\hat{l}_{j}$ identified in step 1 is strong enough to be retained. Define the statistics $\hat{L}_j = \sum\limits_{k \in  \hat{l}_{j}}  (\hat{T}^{(k)}_j)^2$ and this statistics will approximately follow a $\chi^2_{\hat{\kappa}_j}$ distribution with degree of freedom $\hat{\kappa}_j$ under null. Thus we can estimate $\mathcal{\hat{A}}^{[0]}$ by:
\begin{equation}
  \label{eq5}
\mathcal{\hat{A}}^{[0]} = \{j\in [p]; \hat{L}_j \le \varphi^{-1}_{\hat{\kappa}_j}( 1- \alpha_2) \text{ and } \hat{\kappa}_j \neq 0 \}, 
\end{equation}
or equivalently estimate $\mathcal{\hat{A}}^{[1]}$ by:
\begin{equation}
  \label{eq6}
\mathcal{\hat{A}}^{[1]} = \{j\in [p]; \hat{L}_j > \varphi^{-1}_{\hat{\kappa}_j}( 1- \alpha_2) \text{ or } \hat{\kappa}_j = 0 \} , 
\end{equation}
where $\varphi_{\hat{\kappa}_j}$ is the CDF of chi-square distribution with degree of freedom equal to $\hat{\kappa}_j$ and $\alpha_{2}$ the pre-specified significance level.

The second step takes the sum of squares of $\hat{T}^{(k)}_j$ from studies with potential zero correlation as the test statistics. For each feature $j$, we test if $\sum\limits_{k \in  \hat{l}_{j}}  (\hat{T}^{(k)}_j)^2  > \varphi^{-1}_{\hat{\kappa}_j}( 1- \alpha_2)$. If rejected, we conclude that the aggregate effect is strong and the feature needs to be retained, otherwise, we will screen it out. This step performs a second check in addition to the individual testing in step 1 and potentially saves those important features with weak signals in individual studies but strong aggregate effect. 

\begin{table}[t!] 
\caption{Toy example to demonstrate the strength of two-step screening procedure.}
\label{tab:1}\par
\vskip .2cm
\begin{center}
\begin{tabular}{ |c|c|c|c|c| }  
\hline
&  & S1 (signal) & S2 (signal)  & N1 (noise) \\
 \hline 
&  k=1 & $|\hat{T}^{(1)}_1| = 3.71$  & $|\hat{T}^{(1)}_2| = 3.70$  & $|\hat{T}^{(1)}_3| = 0.42$   \\ 
 \hline 
&  k=2 & $|\hat{T}^{(2)}_1| = 3.16$  & $|\hat{T}^{(2)}_2| = 2.71$ &  $|\hat{T}^{(2)}_3| = 0.54$   \\ 
  \hline 
& k=3 & $|\hat{T}^{(3)}_1| = 3.46$  & $|\hat{T}^{(3)}_2| = 2.65$ &  $|\hat{T}^{(3)}_3| = 0.56$   \\ 
 \hline 
& k=4 & $|\hat{T}^{(4)}_1| = 3.63$  & $|\hat{T}^{(4)}_2| = 2.68$ &  $|\hat{T}^{(4)}_3| = 0.12$   \\ 
 \hline 
& k=5 & $|\hat{T}^{(5)}_1| = 3.24$  & $|\hat{T}^{(5)}_2| = 1.94$ &  $|\hat{T}^{(5)}_3| = 0.69$   \\ 
   \hline 
\multirow{5}{*}{TSA-SIS} & $\hat{l}_j$ & $\emptyset$ & $\{2,3,4,5\}$ &  $\{1,2,3,4,5\}$ \\
& $\hat{\kappa}_j$ & $0$ & $4$ & $5$ \\
& $\hat{L}_j$ & - & $25.31 > \varphi_{4}(0.95) $  &  $1.27 <  \varphi_{5}(0.95) $\\
& $\mathcal{\hat{A}}^{[0]}$ & N & N &  Y \\
 & $\mathcal{\hat{A}}^{[1]}$ & Y & Y &  N \\
 \hline
 \multirow{2}{*}{OneStep-SIS} & $\mathcal{\hat{A}}^{[0]}$ & N & Y  & Y \\
 & $\mathcal{\hat{A}}^{[1]}$ & Y & N (FN)  & N \\
 \hline
\end{tabular}
\end{center}
\end{table}

In Table \ref{tab:1}, we use a toy example to demonstrate our idea and compare the two approaches (``OneStep-SIS" vs.  ``TSA-SIS"). In this example, suppose we have five studies ($K=5$) and three features (two signals and one noise). ``S1" is a strong signal with $\beta=0.8$ in all studies, ``S2" is a weak signal with $\beta=0.4$ in all studies and ``N1" is a noise with $\beta=0$. In hypothesis testing, both small $\beta$ and zero $\beta$ can give small marginal correlation and are sometimes indistinguishable. Suppose $T=3.09$ is used as the threshold (corresponding to $\alpha_1=0.001$). For the strong signal ``S1", all studies have large marginal correlations, so both ``OneStep-SIS" and ``TSA-SIS" procedures include it correctly. For the weak signal ``S2", since in many studies it has small correlations, it is incorrectly screened out by ``OneStep-SIS" procedure (False Negative). However, the ``TSA-SIS" procedure saves it in the second step (with $\alpha_2=0.05$). For the noise ``N1", both methods tend to remove it after screening. 

\setcounter{equation}{0} 
\section{Theoretical properties}
\label{sec:theory}

\subsection{Assumptions and conditions}
\label{subsec:assumption}

We impose the following conditions to establish the model selection consistency of our procedure:
\begin{itemize}
\item[(C1)] (Sub-Gaussian Condition) There exist some constants $M_1>0$ and $\eta >0$ such that for all $|t|\le \eta$, $j\in [p]$, $k\in [K]$:  $$E\{\exp(tZ_j^{(k)2})\} \le M_1,  \quad E\{\exp(tW^{(k)2})\} \le M_1.$$ In addition, there exist some $\tau_0 >0 $ such that $\min\limits_{j,k} \theta^{(k)}_j \ge \tau_0$.
\item[(C2)] The number of studies $K= O(p^b)$ for some constant $b\ge 0$. The dimension satisfies: $\log^3(p) = o(n)$ and $\kappa_j \log^2 p = o(n)$, where $\kappa_j$ is defined next.
\item[(C3)] For $j\in \mathcal{A}^{[0]}$, $l_j(j\in \mathcal{A}^{[0]} )= \{k; \rho_j^{(k)} =0  \}$ and $ \kappa_j = |l_j|$. If $k\notin l_j$, then $|\rho_j^{(k)}| \ge C_3 \sqrt{\frac{\log p}{n}} \sqrt{1.01 \theta_j^{(k)}} $, where $C_3 = 3(L+1+b)$.  
\item[(C4)] For $j \in \mathcal{A}^{[1]}$, $l_j(j\in \mathcal{A}^{[1]} )= \{k; |\rho_j^{(k)}| < C_1\sqrt{\frac{\log p}{n}} \sqrt{0.99 \theta_j^{(k)}}  \}$ and $\kappa_j = |l_j|$, where $C_1 = L+1+b$. If $k\notin l_j$, then $|\rho_j^{(k)}| \ge C_3 \sqrt{\frac{\log p}{n}} \sqrt{1.01 \theta_j^{(k)}} $. In addition, we require $\sum\limits_{k \in l_{j}}|\rho_j^{(k)}|^2  \ge \frac{C_{2}(\log^2p + \sqrt{\kappa_j \log p}) }{n}$, where $C_{2}$ is some large positive constant. 
\end{itemize}

The first condition (C1) assumes that each standardized variable $Z_j^{(k)}$ or $W^{(k)}$, $j\in[p]$, $k\in[K]$, marginally follow a sub-Gaussian distribution in each study. This condition relaxes the normality assumption in \citep{fan2008sure,buhlmann2010variable}. The second part of (C1) assumes there always exist some positive $\tau_0$ not greater than the minimum variance of $Z^{(k)}_j W^{(k)}$. In particular, if $(X_j^{(k)} , Y^{(k)})$ jointly follows a multivariate normal distribution, then $\theta^{(k)}_j = 1 + \rho^{(k)2}_j \ge 1$, so we can always pick $\tau_0 = 1$.  

The second condition (C2) allows the dimension $p$ to grow at an exponential rate of sample size $n$, which is a fairly standard assumption in high-dimensional analysis. Many sure screening methods like ``SIS", ``DC-SIS" and ``TPC" have used this assumption \citep{fan2008sure,li2012feature,li2017variable}. Though the PC-simple algorithm \citep{buhlmann2010variable} assumes a polynomial growth of $p_n$ as a function of $n$, we notice that it can be readily relaxed to an exponential of $n$ level.  Further, we require the product $\kappa_j \log^2 p$ to be small, which is used to control the errors in the second step of our screening procedure. It is always true if $K\log^2 p = o(n)$. 

Conditions (C3) assumes a lower bound on non-zero correlation (i.e. $k\notin l_j$) for features from $\mathcal{A}^{[0]}$. In other words, if the marginal correlation $|\rho^{(k)}_{j}|$ is not zero, then it must have a large enough marginal correlation to be detected. While this has been a key assumption for a single study in many sure screening methods \citep{fan2008sure,buhlmann2010variable,li2012feature,li2017variable}, we only impose this assumption for $j\in\mathcal{A}^{[0]}$ rather than all $j\in[p]$. This condition is used to control for type II error in step 1 for features from $\mathcal{A}^{[0]}$. 

Condition (C4) gives assumptions on features from $\mathcal{A}^{[1]}$. We assume the correlations to be small for those $k\in l_j$ and large for those $k\notin l_j$ so that studies with strong or weak signals can be well separated in the first step. This helps control the type II error in step 1 for features from $\mathcal{A}^{[1]}$. For those studies in $l_j$, we further require their sum of squares of correlations to be greater than a threshold, so that type II error can be controlled in step 2. This condition is different from other methods with single study scenario, where they usually assume a lower bound on each marginal correlation for features from $\mathcal{A}^{[1]}$ just like (C3). We relax this condition and only put restriction on their $L_2$ norm. This allows features from $\mathcal{A}^{[1]}$ to have weak signals in each study but combined strong signal. To appreciate this relaxation, we compare the minimal requirements with and without step 2. For each $j\in\mathcal{A}^{[1]}$, in order to detect this feature, we need $|\rho_j^{(k)}| \ge C (\log p/n)^{1/2}$ with some large constant $C$ for all $k\in l_j$, and thus at least $\sum\limits_{k \in l_{j}}|\rho_j^{(k)}|^2  \ge C^2\kappa_j\log p/n$. In comparison, the assumption in (C4) is much weaker in reasonable settings $\kappa_j>>\log p$.

\subsection{Consistency of the two-step screening procedure}
\label{subsec:consist}

We state the first theorem involving the consistency of screening in our step 1: 

 \begin{theorem}
 \label{th1}
Consider a sequence of linear models as in (\ref{eq1}) which satisfy assumptions and conditions (C1)-(C4), define the event $A = \{ \hat{l}_j = l_j \text{ for all } j \in [p] \}$, there exists a sequence $ \alpha_1 = \alpha_1(n,p) \rightarrow 0$ as $(n,p) \rightarrow \infty$ where $\alpha_1 = 2\{1 - \Phi(\gamma \sqrt{\log p})\}$ with $\gamma= 2(L+1+b)$ such that:
\begin{equation}
\label{eq7}
P(A) = 1 - O(p^{-L}) \rightarrow 1 \text{ as } (n,p) \rightarrow \infty.
\end{equation} 
\end{theorem}

The proof of Theorem \ref{th1} can be found in Section 9. This theorem states that the screening in our first step correctly identifies the set $l_j$ for features in both $\mathcal{A}^{[0]}$ and $\mathcal{A}^{[1]}$ (in which strong and weak signals are well separated) and the chance of incorrect assignment is low. Given the results in Theorem \ref{th1}, we can now show the main theorem for the consistency of the two-step screening procedure:

 \begin{theorem}
 \label{th2}
Consider a sequence of linear models as in (\ref{eq1}) which satisfy assumptions and conditions (C1)-(C4), we know there exists a sequence $ \alpha_1 = \alpha_1(n,p) \rightarrow 0$ and $ \alpha_2 = \alpha_2(n,p) \rightarrow 0$ as $(n,p) \rightarrow \infty$ where $\alpha_1 = 2\{1 - \Phi(\gamma \sqrt{\log p})\}$ with $\gamma= 2(L+1+b)$ and $\alpha_2 = 1 - \varphi_{\kappa_j}(\gamma_{\kappa_j} ) $ with $\gamma_{\kappa_j} = \kappa_j + C_4 (\log^2 p + \sqrt{\kappa_j \log p})$ and some constant $C_4>0$ such that:
\begin{equation}
\label{eq8}
P\{\hat{\mathcal{A}}^{[1]}(\alpha_{1}, \alpha_{2}) =\mathcal{A}^{[1]} \} = 1 - O(p^{-L}) \rightarrow 1 \text{ as } (n,p) \rightarrow \infty. 
\end{equation} 
\end{theorem}

The proof of Theorem \ref{th2} can be found in Section \ref{sec:proof}. The result shows that the two-step screening procedure enjoys the model selection consistency and identifies the model specified in (\ref{eq3}) with high probability. The choice of significance level that yields consistency is $\alpha_1 = 2\{1 - \Phi(\gamma \sqrt{\log p})\}$ and $\alpha_2 = 1 - \varphi_{\kappa_j}(\gamma_{\kappa_j} ) $ .

\subsection{Partial faithfulness and Sure screening property}
\label{subsec:sure}

\cite{buhlmann2010variable} first came up with the partial faithfulness assumption which theoretically justified the use of marginal correlation or partial correlation in screening as follows:  
\begin{equation}
  \label{eq9}
\rho_{j\mid S}=0 \text{ for some } S \subseteq  \{j\}^C \text{ implies } \beta_j = 0,
\end{equation}
where $S$ is the set of variables conditioned on. For independence screening, $S = \emptyset$. 

Under the two conditions: the positive definiteness of $\Sigma_X$ and non-zero regression coefficients being realization from some common absolutely continuous distribution, they showed that partial faithfulness held almost surely (Theorem 1 in \cite{buhlmann2010variable}). Since the random effect model described in Section \ref{sec:model} also satisfies the two conditions, the partial faithfulness holds almost surely in each study. 

Thus, we can readily extend their Theorem 1 to a scenario with multiple studies:
\begin{corollary}
 \label{cor1}
Consider a sequence of linear models as in (\ref{eq1}) satisfying the partial faithfulness condition in each study and true active and inactive set defined in (\ref{eq2}), then the following holds for every $j\in [p]$:
\begin{equation}
  \label{eq10}
  \rho_{j\mid S}^{(k)}=0 \text{ for some } k \text{ for some } S \subseteq  \{j\}^C \text{ implies } \beta_j = 0.
\end{equation}
\end{corollary}
The proof is straightforward and thus omitted: if $\rho_{j\mid S}^{(k)}=0$ for some study $k$, then with partial faithfulness, we will have $\beta_j^{(k)} = 0$ for that particular $k$. Since we only consider features with zero or non-zero $\beta_j^{(k)}$'s in all studies in (\ref{eq2}), we will have $\beta_j = 0$. In the case of independence screening (i.e. $S = \emptyset$), $\rho_{j}^{(k)}=0$ for some $k$ will imply a zero $\beta_j$. 

With the model selection consistency in Theorem \ref{th2} and the extended partial faithfulness condition in Corollary \ref{cor1}, the sure screening property of our two-step screening procedure immediately follows:

\begin{corollary}
 \label{cor2}
Consider a sequence of linear models as in (\ref{eq1}) which satisfy assumptions and conditions (C1)-(C4) as well as the extended partial faithfulness condition in Corollary \ref{cor1}, there exists a sequence $ \alpha_1 = \alpha_1(n,p) \rightarrow 0$ and $ \alpha_2 = \alpha_2(n,p) \rightarrow 0$ as $(n,p) \rightarrow \infty$ where $\alpha_1 = 2\{1 - \Phi(\gamma \sqrt{\log p})\}$ with $\gamma= 2(L+1+b)$ and $\alpha_2 = 1 - \varphi_{\kappa_j}(\gamma_{\kappa_j} ) $ with $\gamma_{\kappa_j} = \kappa_j + C_4 (\log^2 p + \sqrt{\kappa_j \log p})  $ such that:
\begin{equation}
\label{eq11}
P\{\mathcal{A} \subseteq \hat{\mathcal{A}}^{[1]}(\alpha_{1}, \alpha_{2}) \} = 1 - O(p^{-L}) \rightarrow 1 \text{ as } (n,p) \rightarrow \infty. 
\end{equation} 
\end{corollary}

The proof of this Corollary simply combines the results of Theorem \ref{th2} and the extended partial faithfulness and is skipped here. 

\setcounter{equation}{0} 
\section{Algorithms for variable selection with multiple studies}
\label{sec:algorithm}

Usually, performing sure screening once may not remove enough unimportant features. In our case since there are multiple studies, we expect our two-step screening procedure to remove many more unimportant features than in single study. If the dimension is still high after applying our screening procedure, we can readily extend the two-step screening procedure to an iterative variable selection algorithm by testing the partial correlation with gradually increasing size of the conditional set $S$. Since such method is a multiple study extension of the PC simple algorithm in \cite{buhlmann2010variable}, we call it ``Multi-PC" algorithm (Section \ref{subsec:multi}).

On the other hand, if the dimension has already been greatly reduced with the two-step screening, we can simply add a second stage group-based feature selection techniques to select the final set of variables (Section \ref{subsec:two}).

\subsection{Multi-PC algorithm}
\label{subsec:multi}

We start from $S=\emptyset$, i.e., our two-step screening procedure and build a first set of candidate active variables: 
\begin{equation}
  \label{eq11}
\mathcal{\hat{A}}^{[1,1]}  = \mathcal{\hat{A}}^{[1]} = \{j\in [p]; \hat{L}_j > \varphi^{-1}_{\hat{\kappa}_j}( 1- \alpha_2) \text{ or } \hat{\kappa}_j = 0 \}. 
\end{equation}

We call this set $stage_1$ active set, where the first index in $[,]$ corresponds to the stage of our algorithm and the second index corresponds to whether the set is for active variables ($[,1]$) or inactive variables ($[,0]$). If the dimensionality has already been decreased by a large amount, we can directly apply group-based feature selection methods such as group lasso to the remaining variables (to be introduced in Section \ref{subsec:two}). 

However, if the dimension is still very high, we can further reduce dimension by increasing the size of $S$ and considering partial correlations given variables in $\mathcal{\hat{A}}^{[1,1]}$. We follow the similar two-step procedure but now using partial correlation of order one instead of marginal correlation and yield a smaller $stage_2$ active set: 
\begin{equation}
  \label{eq12}
 \mathcal{\hat{A}}^{[2,1]} = \{j \in \mathcal{\hat{A}}^{[1,1]}; \hat{L}_{j\mid q} > \varphi^{-1}_{\hat{\kappa}_{j\mid q}}( 1 - \alpha_2) \text{ or } \hat{\kappa}_{j\mid q} = 0, \text{ for all } q\in \mathcal{\hat{A}}^{[1,1]} \backslash \{j\} \} ,
 \end{equation} 
where each self-normalized estimator of partial correlation can be computed by taking the residuals from regressing over the variables in the conditional set. 

We can continue screening high-order partial correlations, resulting in a nested sequence of $m$ active sets:
\begin{equation}
  \label{eq14}
\hat{\mathcal{A}}^{[m,1]} \subseteq \ldots \subseteq \hat{\mathcal{A}}^{[2,1]} \subseteq \hat{\mathcal{A}}^{[1,1]}. 
 \end{equation}  

Note that the active and inactive sets at each stage are non-overlapping and the union of active and inactive sets at a stage $m$ will be the active set in a previous stage $m-1$, i.e., $\mathcal{\hat{A}}^{[m,1]}\cup \mathcal{\hat{A}}^{[m,0]}= \mathcal{\hat{A}}^{[m-1,1]}$. This is very similar to the original PC-simple algorithm, but now at each order-level, we perform the two-step procedure. The algorithm can stop at any stage $m$ when the dimension of $\mathcal{\hat{A}}^{[m,1]}$ already drops to low to moderate level and other common group-based feature selection techniques can be used to select the final set. Alternatively, we can continue the algorithm until the candidate active set does not change anymore. The algorithm can be summarized as follows: 

\noindent\makebox[\linewidth]{\rule{18cm}{0.3pt}}
Algorithm 1. Multi-PC algorithm for variable selection. \\
\noindent\makebox[\linewidth]{\rule{18cm}{0.3pt}}
\begin{itemize}
 \item[Step 1.] Set $m=1$, perform the two-step screening procedure to construct $stage_1$ active set: 
 $$\mathcal{\hat{A}}^{[1,1]} = \{j\in [p]; \hat{L}_j > \varphi^{-1}_{\hat{\kappa}_j}( 1- \alpha_2) \text{ or } \hat{\kappa}_j = 0 \}.  $$
 \item[Step 2.] Set $m=m+1$. Construct the $stage_m$ active set:
 \begin{equation*}
  \begin{split}
   \mathcal{\hat{A}}^{[m,1]} = \{j\in \mathcal{\hat{A}}^{[m-1,1]}; \hat{L}_{j|S} > \varphi^{-1}_{\hat{\kappa}_{j|S}}(1-\alpha_2) \text{ or } \hat{\kappa}_{j|S} = 0, \\ \text{ for all } S \subseteq \mathcal{\hat{A}}^{[m-1,1]} \backslash \{j\} \text{ with } |S| = m -1 \}. 
   \end{split}
  \end{equation*} 
  \item[Step 3.] \textbf{Repeat} Step 2 until $m=\hat{m}_{reach}$, where $\hat{m}_{reach} = \min\{m: |\mathcal{\hat{A}}^{[m,1]}| \le m\}$.
\end{itemize}
\noindent\makebox[\linewidth]{\rule{18cm}{0.3pt}}
\vspace{0.3cm}

\subsection{Two-stage feature selection}
\label{subsec:two}
As an alternative to ``Multi-PC" algorithm for variable selection, we also introduce here a two-stage feature selection algorithm by combining our two-step screening procedure and other regular feature selection methods together. In single study, for example, Fan \& Lv (2008) performed sure independence screening in the first stage followed by model selection techniques including Adaptive Lasso, Dantzig Selector and SCAD, etc., and named those procedures as ``SIS-AdaLasso",``SIS-DS", ``SIS-SCAD" , accordingly. 

In our case, since the feature selection is group-based, we adopt a model selection technique using group Lasso penalty in the second stage:
\begin{equation}
  \label{eq15}
  \min_{\beta} \sum\limits_{k=1}^K ||y^{(k)} - X^{(k)}_{\mathcal{\hat{A}}^{[1]}} \beta^{(k)}_{\mathcal{\hat{A}}^{[1]}} ||^2_2  + \lambda \sum\limits_{j \in \mathcal{\hat{A}}^{[1]} }  ||\beta_j ||_2 \quad , 
\end{equation}
where $\mathcal{\hat{A}}^{[1]}$ is the active set identified from our two-step screening procedure and the tuning parameter $\lambda$ can be chosen by cross-validation or BIC in practice just like for a regular group Lasso problem. We call such two-stage feature selection algorithm as ``TSA-SIS-groupLasso".  

In addition, at any stages of the ``Multi-PC" algorithm when the dimension has already been dropped to a moderate level, the group Lasso-based feature selection techniques can always take over to select the final set of variables. 

\setcounter{equation}{0} 
\section{Numerical evidence}
\label{sec:simulation}

In this section, we demonstrate the advantage of TSA-SIS procedure in comparing to the multiple study extension of SIS (named ``Min-SIS"), which ranks the features by the minimum absolute correlation among all studies. We simulated data according to the linear model in (\ref{eq1}) including $p$ covariates with zero mean and covariance matrix $\Sigma^{(k)}_{i,j}=r^{|i-j|}$ where $\Sigma^{(k)}_{i,j}$ denotes the $(i,j)$th entry of $\Sigma^{(k)}_X$. 

In the first part of simulation, we fixed the sample size $n=100$, $p=1000$, the number of studies $K=5$ and performed $B=1000$ replications in each setting. We assumed that the true active set consisted of only ten variables and all the other variables had zero coefficients (i.e., $s_0=10$). The indices of non-zero coefficients were evenly spaced between $1$ and $p$. The variance of the random error term in linear model was fixed to be $0.5^2$. We randomly drew $r$ from $\{0,0.2,0.4,0.6\}$ and allowed different $r$'s in different studies. We considered the following four settings: 
\begin{enumerate}
\item Homogeneous weak signals across all studies: nonzero $\beta_j$ generated from $\text{Unif}(0.1,0.3)$ and $\beta_j^{(1)}=\beta_j^{(2)}= \ldots=\beta_j^{(K)}=\beta_j$. 
\item Homogeneous strong signals across all studies: nonzero $\beta_j$ generated from $\text{Unif}(0.7,1)$ and $\beta_j^{(1)}=\beta_j^{(2)}= \ldots=\beta_j^{(K)}=\beta_j$. 
\item Heterogeneous weak signals across all studies: nonzero $\beta_j$ generated from $\text{Unif}(0.1,0.3)$ and $\beta_j^{(k)}\sim N(\beta_j, 0.5^2)$.
\item Heterogeneous strong signals across all studies: nonzero $\beta_j$ generated from $\text{Unif}(0.7,1)$ and $\beta_j^{(k)}\sim N(\beta_j, 0.5^2)$.
\end{enumerate}

We evaluated the performance of Min-SIS using receiver operating characteristic (ROC) curves
which measured the accuracy of variable selection independently from the issue of choosing
good tuning parameters (for Min-SIS, the tuning parameter is the top number of features $d$). The OneStep-SIS procedure we mentioned above was actually one special case of the Min-SIS procedure (by thresholding at $\alpha_1$). In presenting our TSA-SIS procedure, we fixed $\alpha_1=0.0001$ and $\alpha_2=0.05$ so the result was just one point on the sensitivity vs. 1-specificity plot. We also performed some sensitivity analysis on the two cutoffs based on the first simulation (see Table \ref{tab:2}) and found the two values to be optimal since they had both high sensitivity and high specificity. Thus we suggested fixing these two values in all the simulations. 

\begin{table}
\caption{Sensitivity analysis on the choice of $\alpha_1$ and $\alpha_2$ in simulation (Sensitivity/Specificity) }
\label{tab:2}\par
\vskip .2cm
\centerline{\tabcolsep=3truept\begin{tabular}{ ccccc }
 \\
 \hline
Sensitivity/Specificity & $\alpha_2 = 0.15$ & 0.05 & 0.01 & 0.001   \\ 
 \hline 
$\alpha_1$=0.01 & 0.793/0.901  & 0.525/0.984 & 0.210/0.999 & 0.142/1.000  \\ 
 0.001 & 0.947/0.826  & 0.864/0.943 & 0.691/0.990 & 0.373/0.999  \\ 
 0.0001 & 0.966/0.816  & 0.922/0.932 & 0.840/0.985 & 0.681/0.998  \\ 
 \hline 
\end{tabular}}
\centerline{{\footnotesize Note: All value are based on average results from $B=1000$ replications.}}
\end{table}

\begin{figure}[h!]
\caption{Simulation results 1-4: the ROC curve is for Min-SIS, the black point is for our TSA-SIS using $\alpha_1=0.0001$ and $\alpha_2=0.05$.}
\label{fig1}
\begin{center}
\includegraphics[scale=0.8]{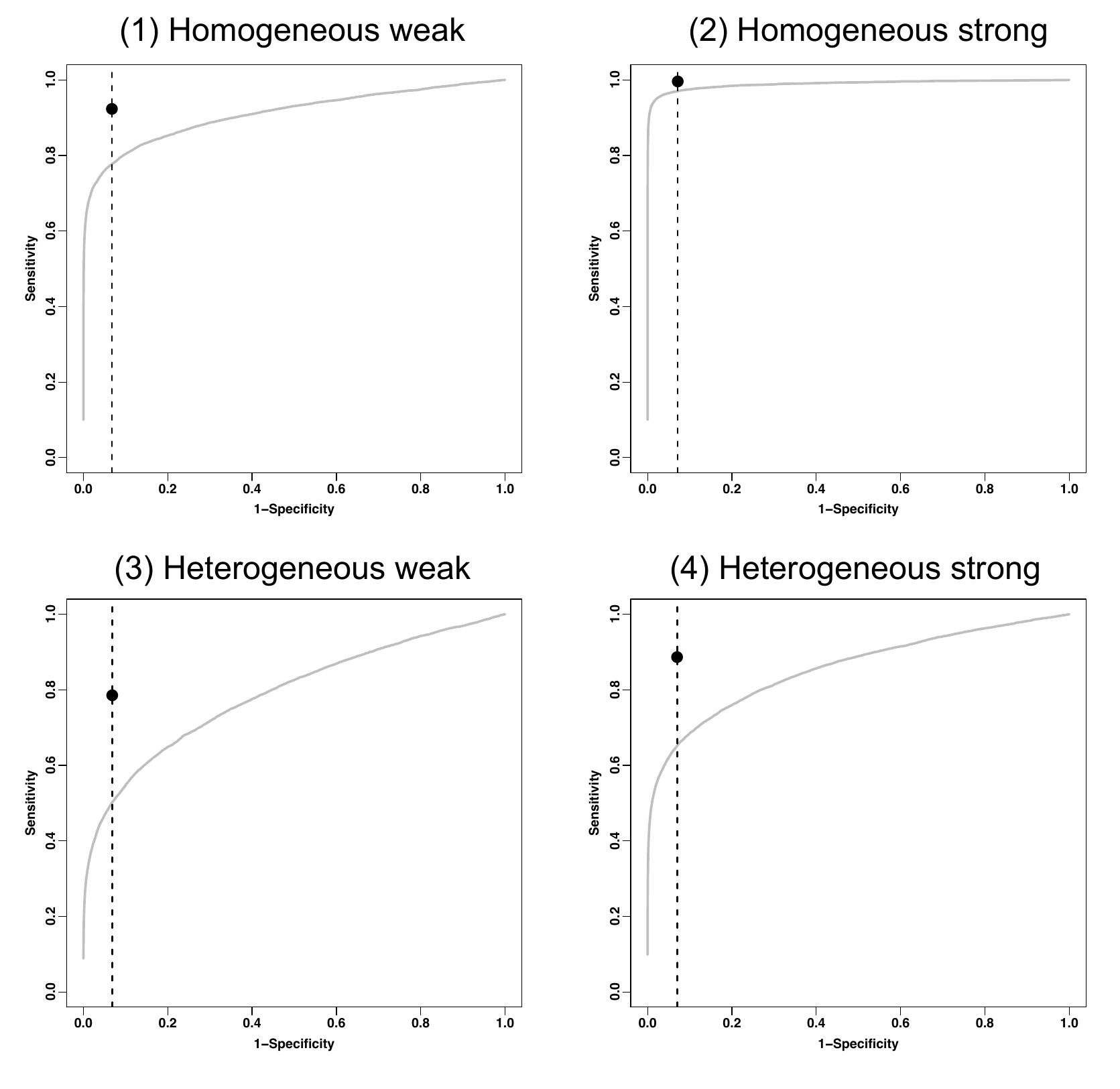}\par
\end{center}
\end{figure}

Figure \ref{fig1} showed the results of simulation 1-4. When the signals were homogeneously weak in all studies as in (1), TSA-SIS clearly outperformed the Min-SIS procedure (above its ROC curve). It reached about $90\%$ sensitivity with controlled false positive errors (specificity $\sim 95\%$). In order to reduce false negatives, Min-SIS had to sacrifice the specificity and increased the false positives, which in the end lost the benefits of performing screening (i.e. end up keeping too many features). When the signals became strong as in (2), both procedures performed equally well. This fit our motivation and theory and showed the strength of our two-step procedure in saving weak signals without much increase in false positive rates. When the signals became heterogeneous as in (3) and (4), both procedures performed worse than before. But the Min-SIS procedure never outperformed the TSA-SIS procedure since it only examined the minimum correlation among all studies while the two-step procedure additionally considered the aggregate statistics. 

\setcounter{equation}{0} 
\section{Real data application}
\label{sec:application}

We next demonstrated our method in three microarray datasets of triple-negative breast cancer (TNBC, sometimes a.k.a. basal-like), an aggressive subtype of breast cancer usually with poor prognosis.  Previous studies have shown that the tumor suppressor protein ``p53" played an important role in breast cancer prognosis and its expression was associated with both disease-free survival and overall survival in TNBC \citep{yadav2015biomarkers}. Our purpose was to identify the genes most relevant and predictive to the response - the expression level of \textit{TP53} gene, which encodes p53 protein. The three datasets are publicly available on authors' website or at GEO repository including METABRIC (a large cohort consisting of roughly 2000 primary breast tumours), GSE25066 and GSE76250 \citep{curtis2012genomic,itoh2014estrogen,liu2016comprehensive}. We subset the data to focus on the TNBC cases only and ended up with 275, 178 and 165 TNBC samples in each dataset, respectively. After routine preprocessing and filtering by including genes sufficiently expressed and with enough variation, a total of 3377 genes remained in common for the analysis. 

We applied our Multi-PC algorithm and compared to the OneStep-SIS procedure as well as the Min-SIS method by using $d=n/\log(n)=49$ (as suggested by their paper). We used $\alpha_1=0.0001$ and $\alpha_2=0.05$ (as determined by sensitivity analysis in simulation) and the ``Multi-PC" algorithm only ran up to the first order (i.e. $m=2$) and stopped with six features. This again showed the power of screening with multiple studies. After feature selection, we fit the linear model in each study to obtain the coefficient estimates and adjusted $R^2$. Table \ref{tab:3} showed the coefficient estimates and standard errors of the final set of six genes selected by our procedure. We added three columns to indicate whether they were also retained by the Min-SIS (and their relative rank) or OneStep-SIS procedures. As we can see from the table, all the six genes selected by our procedure were missed by the other methods. Those genes typically had weak signals in one or more studies thus were very likely to be incorrectly excluded if only one step screening is performed. Since the METABRIC study had a larger sample size, all the coefficients appeared to be more significant than the other two studies. 

\begin{table}[t!] 
\caption{The six genes selected by our TSA-SIS procedure.}
\label{tab:3}\par
\vskip .2cm
\centerline{\tabcolsep=3truept\begin{tabular}{ ccccccc }  
\\
\hline\hline  
Gene & METABRIC  & GSE25066 & GSE76250 & Min-SIS & Rank in & OneStep-SIS  \\ 
&  Est (SE) & Est (SE) & Est (SE) & $d$=49 & Min-SIS  & $|S|$=25 \\
 \hline 
Intercept & 7.600 (1.502)  & 0.213 (0.553) & -1.783 (0.971) & - & - & -  \\ 
EXOC1 & 0.251 (0.081)$^{\ast\ast}$  & 0.278 (0.157)$^{.}$ & 0.293 (0.167)$^{.}$ & N & 164 & N  \\ 
ITGB1BP1 & -0.134 (0.045)$^{\ast\ast}$  & 0.003 (0.111) & -0.178 (0.194)  & N & 123 & N  \\ 
RBM23 & 0.168 (0.078)$^{\ast}$  & 0.144 (0.167) & 0.367 (0.168)$^{\ast}$ & N & 152 & N  \\ 
SETD3 & -0.166 (0.081)$^{\ast}$  & 0.366 (0.184)$^{\ast}$ & -0.080 (0.175) & N & 101 & N  \\ 
SQSTM1 & -0.114 (0.050)$^{\ast}$  & 0.029 (0.099) & 0.245 (0.183) & N & 98 & N  \\ 
TRIOBP & -0.126 (0.062)$^{\ast}$  & 0.084 (0.118) & 0.628 (0.261)$^{\ast}$ & N & 91 & N  \\ 
   \hline 
   Adjusted-$R^2$ & 0.151 & 0.522 & 0.359 & & &  \\
 \hline\hline 
\end{tabular}}
{\footnotesize Note: ``$.$" indicates significant level of 0.1, ``$\ast$" for level of 0.05, ``$\ast\ast$" for level of 0.01.}
\end{table}

The gene \textit{EXOC1} and p53 are both components of the Ras signaling pathway which is responsible for cell growth and division and can ultimately lead to cancer \citep{rajalingam2007ras}. \textit{RBM23} encodes for an RNA-binding protein implicated in the regulation of estrogen-mediated transcription and has been found to be associated with p53 indirectly via a heat shock factor \citep{asano2016ier5}. \textit{ITGB1BP1} encodes for an integrin protein which is essential for cell adhesion and other downstream signaling pathways that are also modulated by p53 \citep{brakebusch2002integrins}. 

\section{Discussion}
\label{sec:discussion}

In this paper, we proposed a two-step screening procedure for high-dimensional regression analysis with multiple related studies. In a fairly general framework with weaker assumptions on the signal strength, we showed that our procedure possessed the sure screening property for exponentially growing dimensionality without requiring the normality assumption. We have shown through simulations that our procedure consistently outperformed the rank-based SIS procedure independent of their tuning parameter $d$. As far as we know, our paper is the first proposed procedure to perform variable screening in high-dimensional regression when there are multiple related studies. In addition, we also introduced two applicable variable selection algorithms following the two-step screening procedure.

Variable selection in regression with multiple studies have been studied in a subfield of machine learning called multi-task learning (MTL) before and the general procedure is to apply regularization methods by putting group Lasso penalty, fused Lasso penalty or trace norm penalty, etc. \citep{argyriou2007multi,zhou2012modeling,ji2009accelerated}. However, at ultra-high dimension, such regularization methods usually fail due to challenges in computation expediency, statistical accuracy and algorithmic stability. Instead, sure screening can be used as a fast algorithm for preliminary feature selection, and as long as it exhibits comparable statistical performance both theoretically and empirically, its computational advantages make it a good choice in application \citep{genovese2012comparison}. Our method has provided an alternative to target the high-dimensional multi-task learning problems. 

The current two-step screening procedure is based on the linear models but relaxes the Gaussian assumption to sub-Gaussian distribution. One can apply a modified Fisher's z-transformation estimator rather than our self-normalized estimator to readily accommodate general elliptical distribution families \citep{li2017variable}. In biomedical applications, non-continuous outcomes such as categorical, count or survival outcomes are more commonly seen. \cite{fan2010sure} extended SIS and proposed a more general independent learning approach for generalized linear models by ranking the maximum marginal likelihood estimates. \cite{fan2011nonparametric} further extended the correlation learning to marginal nonparametric learning for screening in ultra-high dimensional additive models. Other researchers exploited more robust measure for the correlation screening \citep{zhu2011model,li2012feature,balasubramanian2013ultrahigh}. All these measures can be our potential extension by modifying the marginal utility used in the screening procedure. Besides, the idea of performing screening with multiple studies is quite general and is applicable to relevant statistical models other than the regression model, for example, Gaussian graphical model with multiple studies. We leave these interesting problems in future study.

\setcounter{equation}{0} 
\section{Proofs}
\label{sec:proof}

We start by introducing three technical lemmas that are essential for the proofs of the main results. By the scaling property of $\hat{T}_j^{(k)}$ and Remark \ref{rem:scaling}, without loss of generality, we can assume $E(X^{(k)}_j) = E(Y^{(k)}) = 0$ and $\text{var}(X^{(k)}_j) = \text{var}(Y^{(k)}) = 1$ for all $k\in [K]$, $j\in [p]$. Therefore in the proof we do not distinguish between $\sigma_j^{(k)}$ and  $\rho_j^{(k)}$. The first lemma is on the concentration inequalities of the self-normalized covariance and $\hat{\theta}^{(k)}_j $. 

\begin{lemma}
Under the assumptions (C1) and (C2), for any $\delta \ge 2$ and $M>0$, we have: 
\begin{itemize}
\item[{\normalfont(i)}] $P(\max\limits_{j,k}|\frac{\hat{\sigma}^{(k)}_j  - \sigma^{(k)}_j }{(\hat{\theta}^{(k)}_j)^{1/2} } |  \ge \delta\sqrt{\frac{\log p}{n}} ) = O ((\log p)^{-1/2} p^{-\delta + 1 + b})$,
\item[{\normalfont(ii)}] $P(\max\limits_{j,k}| \hat{\theta}^{(k)}_j - \theta^{(k)}_j | \ge C_{\theta} \sqrt{\frac{\log p}{n}} ) = O (p^{-M})$,
\end{itemize}
where $C_{\theta}$ is a positive constant depending on $M_1$, $\eta$ and $M$ only. 
\label{lemma1}
\end{lemma}

The second and third lemmas, which will be used in the proof of Theorem \ref{th2}, describe the concentration behaviors of $\hat{H}_j^{(k)} := \frac{\frac{1}{\sqrt{n}}\sum\limits_{i=1}^n[ (X^{(k)}_{ij}- \bar{X}^{(k)}_j ) ( Y^{(k)}_i - \bar{Y}^{(k)} ) - \rho^{(k)}_j] }{\sqrt{\theta^{(k)}_j} } = \hat{T}^{(k)}_j \sqrt{\frac{\hat{\theta}^{(k)}_j }{\theta^{(k)}_j}} - \frac{\sqrt{n} \rho^{(k)}_j  }{ \sqrt{\theta^{(k)}_j} }$ and $\check{H}_j^{(k)} := \frac{\frac{1}{\sqrt{n}}\sum\limits_{i=1}^n (X^{(k)}_{ij} Y^{(k)}_i  - \rho^{(k)}_j) }{\sqrt{\theta^{(k)}_j} }$.  

\begin{lemma}
There exists some constant $c>0$ such that,
\begin{equation*}
P(|\sum\limits_{k\in l_j} [\check{H}_j^{(k)2} - 1 ]|> t ) \le 2\exp(-c \min[\frac{t^2}{\kappa_j}, t^{1/2}]), 
\end{equation*}
where $c$ depends on $M_1$ and $\eta$ only.
\label{lemma2}
\end{lemma}

\begin{lemma}
There exists some constant $C_H>0$ such that,
\begin{eqnarray*}
P(\max\limits_{j,k} |\check{H}_j^{(k)} - \hat{H}_j^{(k)}|>  C_H \sqrt{\frac{\log^2 p }{n}} ) &=& O(p^{-M}),\\
P(\max\limits_{j,k} |\check{H}_j^{(k)2} - \hat{H}_j^{(k)2}|>  C_H \sqrt{\frac{\log^3 p }{n}} ) &=& O(p^{-M}),
\end{eqnarray*}
where $C_H$ depends on $M_1$, $\eta$, $M$ and $\tau_0$ only.
\label{lemma3}
\end{lemma}

The proofs of the three lemmas are provided in the Appendix. 

\begin{proof}[Proof of Theorem \ref{th1}] 

We first define the following error events: 
$$E^{I,\mathcal{A}^{[0]}}_{j,k}= \{ |\hat{T}^{(k)}_j| > \Phi^{-1} (1- \alpha_1/2) \text{ and } j \in \mathcal{A}^{[0]}, k\in l_j \}, $$
$$E^{II,\mathcal{A}^{[0]}}_{j,k}= \{ |\hat{T}^{(k)}_j| \le \Phi^{-1} (1- \alpha_1/2) \text{ and } j \in \mathcal{A}^{[0]}, k\notin l_j  \}, $$
$$E^{I,\mathcal{A}^{[1]}}_{j,k}= \{ |\hat{T}^{(k)}_j| > \Phi^{-1} (1- \alpha_1/2) \text{ and } j \in \mathcal{A}^{[1]}, k\in l_j \}, $$
$$E^{II,\mathcal{A}^{[1]}}_{j,k}= \{ |\hat{T}^{(k)}_j| \le \Phi^{-1} (1- \alpha_1/2) \text{ and } j \in \mathcal{A}^{[1]}, k\notin l_j \}. $$
To show Theorem 1 that $P(A) = 1- O(p^{-L})$, it suffices to show that,   
\begin{equation}
\label{proof:eq1}
P\{\bigcup_{j,k}(E^{I,\mathcal{A}^{[0]}}_{j,k} \cup E^{II,\mathcal{A}^{[0]}}_{j,k} ) \} = O(p^{-L}),
 \end{equation}
and 
\begin{equation}
\label{proof:eq2}
P\{\bigcup_{j,k}(E^{I,\mathcal{A}^{[1]}}_{j,k} \cup E^{II,\mathcal{A}^{[1]}}_{j,k} ) \} = O(p^{-L}).
 \end{equation}
One can apply Lemma \ref{lemma1} to bound each component in (\ref{proof:eq1}) and (\ref{proof:eq2}) with $\alpha_1 = 2\{1 - \Phi(\gamma \sqrt{\log p})\}$ and $\gamma= 2(L+1+b)$. Specifically, we obtain that,
\begin{equation} \label{eq:proof1.1.1}
P(\bigcup_{j,k} E^{I,\mathcal{A}^{[0]}}_{j,k}) =  P( \max\limits_{j\in \mathcal{A}^{[0]}, k\in l_j}|\hat{T}^{(k)}_j | \ge \gamma \sqrt{\log p}  )  =  O(\frac{1}{\sqrt{\log p}} p^{-\gamma + 1 + b}) =  o(p^{-L}),
\end{equation}
where the second equality is due to Lemma \ref{lemma1} (i) with $\delta= \gamma$, noting that $\sigma^{(k)}_j=0$ and $\hat{T}^{(k)}_j=\sqrt{n}\hat{\sigma}^{(k)}_j/\sqrt{\hat{\theta}_j^{(k)}}$. In addition, we have that,
\begin{equation} \label{eq:proof1.1.2}
\begin{split}
P(\bigcup_{j,k} E^{I,\mathcal{A}^{[1]}}_{j,k}) = & P\{ \max\limits_{j\in \mathcal{A}^{[1]}, k\in l_j} |\hat{T}^{(k)}_j | \ge \gamma \sqrt{\log p}  \} \\ 
 \le & P(\max\limits_{j\in \mathcal{A}^{[1]}, k\in l_j} |\frac{\hat{\sigma}^{(k)}_j  - \rho^{(k)}_j }{(\hat{\theta}^{(k)}_j)^{1/2} } | \ge (\gamma - C_1) \sqrt{\frac{\log p}{n}} ) + O(p^{-L})\\
 = & O(\frac{1}{\sqrt{\log p}} p^{-(\gamma-C_1) + 1 + b})+O(p^{-L}) \\
 = & O( p^{-L}),
 \end{split}
\end{equation}
where the inequality on the second line is due to assumption (C4) on $l_j$ for $j\in \mathcal{A}^{[1]}$, Lemma \ref{lemma1} (ii) with $M=L$, and assumption (C1) $\min_{j,k}\theta_j^{(k)}\geq \tau_0$, i.e., $\hat{\theta}^{(k)}_j \ge \theta_j^{(k)}-C_{\theta}(\log p /n)^{1/2}\ge 0.99 \theta^{(k)}_j$. The equality on the third line follows from Lemma \ref{lemma1} (i) where $\delta= \gamma - C_1= L+1+b$. In the end, we obtain that,
\begin{equation} \label{eq:proof1.1.3}
\begin{split}
P\{ \bigcup_{j,k} (E^{II,\mathcal{A}^{[0]}}_{j,k} \cup E^{II,\mathcal{A}^{[1]}}_{j,k} )\} = & P( \max\limits_{j, k\notin l_j}|\hat{T}^{(k)}_j | < \gamma \sqrt{\log p}  )  \\ 
 \le & P(\max\limits_{j, k\notin l_j} |\frac{\hat{\sigma}^{(k)}_j  - \rho^{(k)}_j }{(\hat{\theta}^{(k)}_j)^{1/2} } | \ge (C_3 - \gamma) \sqrt{\frac{\log p}{n}} )+O(p^{-L}) \\
 = & O(\frac{1}{\sqrt{\log p}} p^{-(C_3 - \gamma) + 1 + b})+O(p^{-L}) \\
 = & O( p^{-L}),
 \end{split}
\end{equation}
where the inequality on the second line is due to assumptions (C3) and (C4) on $l_j$, Lemma \ref{lemma1} (ii) with $M=L$ and assumption (C1) on sub-Gaussian distributions, i.e., $\hat{\theta}^{(k)}_j \le \theta_j^{(k)}+C_{\theta}(\log p /n)^{1/2} \le 1.01 \theta^{(k)}_j$. In particular, we have implicitly used the fact that $\max_{j,l}\theta_j^{(k)}$ is upper bounded by a constant depending on $M_1$ and $\eta$ only. The equality on the third line follows from Lemma  \ref{lemma1} (i) where $\delta = C_3 - \gamma =L+1+b$.

Finally, we complete the proof by combining (\ref{eq:proof1.1.1})-(\ref{eq:proof1.1.3}) to show (\ref{proof:eq1})-(\ref{proof:eq2}).

\end{proof}

\begin{proof}[Proof of Theorem \ref{th2}] 

We first define the following error events:
$$E^{\mathcal{A}^{[0]},2}_{j} = \{ |\hat{L}_j| > \varphi^{-1} (1- \alpha_2) \text{ or } \hat{\kappa}_j =0\} \text{ for } j \in \mathcal{A}^{[0]}, $$
$$E^{\mathcal{A}^{[1]},2}_{j} = \{ |\hat{L}_j| < \varphi^{-1} (1- \alpha_2) \text{ and } \hat{\kappa}_j \neq 0\} \text{ for } j \in \mathcal{A}^{[1]}. $$

To prove Theorem 2, we only need to show that,
\begin{equation}
\label{proof:eq3}
P(\bigcup_{j\in \mathcal{A}^{[0]} } E^{\mathcal{A}^{[0]},2}_{j} ) = O(p^{-L}) \quad \text{ and }\quad P(\bigcup_{j\in \mathcal{A}^{[1]} } E^{\mathcal{A}^{[1]},2}_{j} ) = O(p^{-L}), 
 \end{equation}
with $\alpha_{2,\kappa_j}:=  1- \varphi_{\kappa_j} [\kappa_j + C_4 (\log^2 p + \sqrt{\kappa_j \log p}) ] := 1 - \varphi_{\kappa_j} (\gamma_{\kappa_j})$. 

Recall the event $A$ defined in Theorem \ref{th1}. Thus we have that, 
\begin{equation*}
\begin{split}
& P\{ (\cup_{j\in \mathcal{A}^{[0]} } E^{\mathcal{A}^{[0]},2}_{j} ) \bigcup  (\cup_{j\in \mathcal{A}^{[1]} } E^{\mathcal{A}^{[1]},2}_{j})  \} \\
\le & P(A^C) + p \max\limits_{j\in \mathcal{A}^{[0]} }  P(\sum\limits_{k\in l_j} \hat{T}^{(k)2}_j > \gamma_{\kappa_j}  ) + p \max\limits_{j\in \mathcal{A}^{[1]}, \kappa_j\neq 0 } P(\sum\limits_{k\in l_j} \hat{T}^{(k)2} <\gamma_{\kappa_j}  ). 
 \end{split}
\end{equation*}

Therefore, given the results in Theorem \ref{th1}, it suffices to show, 
\begin{equation}
\label{proof:eq4}
P(\sum\limits_{k\in l_j} \hat{T}^{(k)2} >\gamma_{\kappa_j}  ) = O(p^{-L-1}) \text{ for any } j\in  \mathcal{A}^{[0]},
 \end{equation}
and 
\begin{equation}
\label{proof:eq5}
P(\sum\limits_{k\in l_j} \hat{T}^{(k)2} < \gamma_{\kappa_j}  ) = O(p^{-L-1}) \text{ for any } j\in  \mathcal{A}^{[1]} \text{ and } \kappa_j >0.
 \end{equation}

We first prove equation (\ref{proof:eq4}). Since $j\in  \mathcal{A}^{[0]}$, we have $\hat{H}_j^{(k)} =  \hat{T}^{(k)}_j \sqrt{\frac{\hat{\theta}^{(k)}_j }{\theta^{(k)}_j}} $. We are ready to bound the probability of $\sum_{k\in l_j} \hat{T}^{(k)2}_j  > \gamma_{\kappa_j}$ below.
\begin{equation*}
\begin{split}
& P(\sum\limits_{k\in l_j} \hat{T}^{(k)2}_j  > \gamma_{\kappa_j} ) \\
\le & P(\sum\limits_{k\in l_j} \hat{H}^{(k)2}_j > (1- \frac{C_{\theta}}{\tau_0} \sqrt{\frac{\log p}{n}} )  \gamma_{\kappa_j}  ) + O(p^{-L-1}) \\ 
\le & P(\sum\limits_{k\in l_j} (\check{H}^{(k)2}_j  - 1) > (1- \frac{C_{\theta}}{\tau_0} \sqrt{\frac{\log p}{n}} )  \gamma_{\kappa_j}   - \kappa_j - \kappa_jC_H \sqrt{\frac{\log^3 p} {n}}  ) + O(p^{-L-1})  \\ 
= & P(\sum\limits_{k\in l_j} (\check{H}^{(k)2}_j  - 1) > \kappa_j +  C_4(\log^2 p + \sqrt{\kappa_j \log p}) - \frac{C_{\theta}}{\tau_0} \sqrt{\frac{\kappa_j^2 \log p} {n}}    \\
    & - \frac{C_{\theta}C_4}{\tau_0} (\sqrt{\frac{\log^5 p}{n}} + \sqrt{\frac{\kappa_j  \log^2 p}{n}}) - \kappa_j - \kappa_jC_H \sqrt{\frac{\log^3 p} {n}}  )  + O(p^{-L-1}) \\
\le & P(\sum\limits_{k\in l_j} (\check{H}^{(k)2}_j  - 1) > C_2'  (\log^2 p + \sqrt{\kappa_j \log p}) ) + O(p^{-L-1}) \\
= & O(p^{-L-1}).
 \end{split}
\end{equation*}

The inequality on the second line is due to assumption (C1) that $\min\limits_{j,k} \theta^{(k)}_j \ge \tau_0 > 0$ and Lemma \ref{lemma1} (ii) with $M=L+1$. The inequality on the third line follows from Lemma \ref{lemma3} with $M=L+1$. The inequality on the fifth line is by the choice of $\gamma_{\kappa_j}$ with a sufficiently large $C_4>0$ and the assumption (C2) that $\log^3 p = o(n)$ and $\kappa_j \log^2 p = o(n)$. The last equality follows from Lemma \ref{lemma2}.

Lastly, we prove (\ref{proof:eq5}) as follows,
\begin{equation}
\label{proof:eq6}
\begin{split}
& P(\sum\limits_{k\in l_j} \hat{T}^{(k)2}_j  < \gamma_{\kappa_j} ) \\
= & P(\sum\limits_{k\in l_j} ( \hat{H}^{(k)}_j  + \frac{\sqrt{n} \rho^{(k)}_j }{\sqrt{\theta^{(k)}_j } })^2  \frac{\theta^{(k)}_j}{\hat{\theta}^{(k)}_j } < \gamma_{\kappa_j}  ) \\ 
\le & P(\sum\limits_{k\in l_j} ( \hat{H}^{(k)}_j  + \frac{\sqrt{n} \rho^{(k)}_j }{\sqrt{\theta^{(k)}_j } })^2 \le (1 + \frac{C_{\theta}}{\tau_0} \sqrt{\frac{\log p}{n}})  \gamma_{\kappa_j} )  + O(p^{-L-1})  \\ 
\le &  P(\sum\limits_{k\in l_j} (\check{H}^{(k)2}_j  - 1) \le \kappa_jC_H \sqrt{\frac{\log^3 p} {n}}  - \kappa_j  + (1 + \frac{C_{\theta}}{\tau_0} \sqrt{\frac{\log p}{n}})  \gamma_{\kappa_j}  - C_m n  \sum\limits_{k\in l_j} \rho^{(k)2}_j  \\ 
&  - 2  \sum\limits_{k\in l_j} \check{H}^{(k)}_j \frac{\sqrt{n} \rho^{(k)}_j }{\sqrt{\theta^{(k)}_j } }  + 2C_H \sqrt{\frac{\log^2 p}{n}} \sum\limits_{k\in l_j} \frac{\sqrt{n} | \rho^{(k)}_j | }{\sqrt{\theta^{(k)}_j } }) + O(p^{-L-1}).  \\  
 \end{split}
\end{equation}

The inequality on the third line is due to assumption (C1) that $\min\limits_{j,k} \theta^{(k)}_j \ge \tau_0 > 0$ and Lemma \ref{lemma1} (ii) with $M=L+1$. The inequality on the fourth line follows from Lemma \ref{lemma3} (both equations) and $\min\limits_{j,k} (\theta^{(k)}_j)^{-1}:=C_m >0 $, guaranteed by the sub-Gaussian assumption in assumption (C1). 

We can upper bound the term $2C_H \sqrt{\frac{\log^2 p}{n}} \sum_{k\in l_j} \frac{\sqrt{n} | \rho^{(k)}_j | }{\sqrt{\theta^{(k)}_j } }$ in (\ref{proof:eq6}) as follow,
\begin{equation}\label{eq:proof2.1.1}
2C_H \sqrt{\frac{\log^2 p}{n}} \sum\limits_{k\in l_j} \frac{\sqrt{n} | \rho^{(k)}_j | }{\sqrt{\theta^{(k)}_j } } \le 2C_H \sqrt{\frac{\log^2 p}{n}} \frac{\sqrt{n}}{\sqrt{\tau_0 } }  \sqrt{\kappa_j} \sqrt{\sum\limits_{k\in l_j} \rho^{(k)2}_j}    = o(\sqrt{n  \sum\limits_{k\in l_j} \rho^{(k)2}_j}    ).
\end{equation}
The first inequality is by the Cauchy-Schwarz inequality and assumption (C1), and the second equality by the assumption (C2) that $\kappa_j \log^2 p = o(n)$. 

We next upper bound the term $-2 \sum_{k\in l_j} \check{H}^{(k)}_j \frac{\sqrt{n} \rho^{(k)}_j }{\sqrt{\theta^{(k)}_j } }$ with high probability. Note that $\theta_j^{(k)}$ is bounded below and above, i.e., $\tau_0\leq\theta_j^{(k)}\leq C_m^{-1}$ by assumption (C1). In addition, $\check{H}^{(k)}_j$ has zero mean and is sub-exponential with bounded constants by assumption (C1). By Bernstein inequality (Proposition 5.16 in \cite{vershynin2010introduction}), we have with some constant $c'>0$,
 \begin{equation*}
P(|2\sum\limits_{k\in l_j} | \check{H}^{(k)}_j \frac{\sqrt{n} | \rho^{(k)}_j | }{\sqrt{\theta^{(k)}_j }} | >t  ) \le 2\exp(-c' \min [\frac{t^2}{n \sum\limits_{k\in l_j} \rho^{(k)2}_j }], \frac{t}{\max\limits_{k\in l_j} \sqrt{n} | \rho^{(k)}_j|  }     ).
\end{equation*}
We pick $t= C_B \sqrt{ n  \sum_{k\in l_j} \rho^{(k)2}_j \log^2 p   }$ with a large constant $C_B$ in the inequality above and apply (\ref{eq:proof2.1.1}) to reduce (\ref{proof:eq6}) as follows,
\begin{equation*}
\begin{split}
& P(\sum\limits_{k\in l_j} \hat{T}^{(k)2}_j  < \gamma_{\kappa_j} ) \\
 \le & P(\sum\limits_{k\in l_j} (\check{H}^{(k)2}_j  - 1) \le - C_m n  \sum\limits_{k\in l_j} \rho^{(k)2}_j +  2C_B \sqrt{ n  \sum\limits_{k\in l_j} \rho^{(k)2}_j \log^2 p}   \\
 & + 2C_4  \sqrt{\kappa_j \log p} + 2C_4 \log^2 p )  + O(p^{-L-1})  \\
 \le & P(\sum\limits_{k\in l_j} (\check{H}^{(k)2}_j  - 1) \le - C_m C_2  (\log^2 p + \sqrt{\kappa_j \log p} )  +  2C_B \sqrt{ C_2 \log^2 p (\log^2 p  + \sqrt{\kappa_j \log p}) }  \\ 
& +  2C_4  \sqrt{\kappa_j \log p} + 2C_4 \log^2 p )  + O(p^{-L-1})  \\
 \le & P(\sum\limits_{k\in l_j} (\check{H}^{(k)2}_j  - 1) \le -C'_2 (\log^2 p + \sqrt{\kappa_j \log p}) ) +  O(p^{-L-1})  \\
 = & O(p^{-L-1}) .
 \end{split}
\end{equation*}

The inequality on the first line is obtained by the choice of $\gamma_{\kappa_j}$ with the chosen $C_4>0$ and the assumption (C2) that $\kappa_j \log^2 p = o(n)$. The inequalities on the second line and third line are by the assumption (C4) that $\sum_{k \in l_{j}}|\rho_j^{(k)}|^2  \ge \frac{C_2(\log^2p + \sqrt{\kappa_j \log p}) }{n}$ for a sufficiently large $C_2>0$. The last equality is by Lemma \ref{lemma2}. 

This completes the proof of (\ref{proof:eq4}) and (\ref{proof:eq5}), which further yields to
$$P\{ (\cup_{j\in \mathcal{A}^{[0]} } E^{\mathcal{A}^{[0]},2}_{j} ) \bigcup  (\cup_{j\in \mathcal{A}^{[1]} } E^{\mathcal{A}^{[1]},2}_{j})  \}  = O(p^{-L}),  $$
with the results from Theorem \ref{th1}. Therefore we complete the proof of Theorem \ref{th2}. 

\end{proof}

\renewcommand{\theequation}{A\arabic{equation}}    
\setcounter{equation}{0}  
\section*{\centerline{Appendix}}  

\subsection*{S1. Proof of Lemma 1}
\label{subsec:lemma1}

\begin{proof} 
Part (i) immediately follows from Lemma 2 (i) equation (25) in \cite{cai2011adaptive}. To prove part (ii), we need to bound the three terms on the right side of the following inequality,
\begin{equation}
\label{supp:eq1}
\max\limits_{j,k}| \hat{\theta}^{(k)}_j - \theta^{(k)}_j | \le \max\limits_{j,k}| \hat{\theta}^{(k)}_j - \tilde{\theta}^{(k)}_j | + \max\limits_{j,k} | \tilde{\theta}^{(k)}_j - \check{\theta}^{(k)}_j | + \max\limits_{j,k} | \check{\theta}^{(k)}_j - \theta^{(k)}_j |, 
\end{equation}
where $\tilde{\theta}^{(k)}_j := \frac{1}{n}\sum\limits_{i=1}^n ( X^{(k)}_{ij}  Y^{(k)}_i - \tilde{\rho}^{(k)}_j)^2$ with $\tilde{\rho}^{(k)}_j = \frac{1}{n}\sum\limits_{i=1}^n X^{(k)}_{ij} Y^{(k)}_i $, and $\check{\theta}^{(k)}_j := \frac{1}{n}\sum\limits_{i=1}^n ( X^{(k)}_{ij} Y^{(k)}_i  - \rho^{(k)}_j)^2$. Note that $E(\check{\theta}^{(k)}_j ) = \theta^{(k)}_j.$

By the marginal sub-Gaussian distribution assumption in assumption (C1), we have that $( X^{(k)}_{ij}  Y^{(k)}_i - \rho^{(k)}_j )^2 $ has mean $\theta^{(k)}_j$ and finite Orlicz $\psi_{1/2}$-norm (see, e.g., \cite{adamczak2011restricted}). Thus we can apply equation (3.6) of \cite{adamczak2011restricted}, i.e.,
\begin{equation*}
P(\max\limits_{j,k} \sqrt{n}|\check{\theta}^{(k)}_j -  \theta^{(k)}_j | > t  ) \le  2\exp(-c\min[\frac{t^2}{n}, t^{1/2}]), 
\end{equation*}
with $t=(C_{\theta}/3) \sqrt{n\log p}$ for a large enough constant $C_{\theta}>0$ depending on $M_1, \eta$ and $M$ only to obtain that,
\begin{equation}
\label{supp:eq2}
P(\max\limits_{j,k} |\check{\theta}^{(k)}_j -  \theta^{(k)}_j | > (C_{\theta}/3)\sqrt{\frac{\log p}{n}} ) =O(p^{-M}). 
\end{equation}
We have used the assumption $\log p = o(n^{1/3})$ in assumption (C2) to make sure $\frac{t^2}{n} \leq t^{1/2}$.

By applying equation (1) in supplement of \cite{cai2011adaptive}, we obtain that,
\begin{equation}
\label{supp:eq3}
P(\max\limits_{j,k} |\tilde{\theta}^{(k)}_j -  \hat{\theta}^{(k)}_j | > (C_{\theta}/3)\sqrt{\frac{\log p}{n}} ) =O(p^{-M}). 
\end{equation}
In addition, by a similar truncation argument as that in the proof of Lemma 2 in \cite{cai2011adaptive} and equation (7) therein, we obtain that by picking a large enough $C_{\theta}>0$, 
\begin{equation}
\label{supp:eq4}
P(\max\limits_{j,k} |\tilde{\theta}^{(k)}_j -  \check{\theta}^{(k)}_j | > (C_{\theta}/3)\sqrt{\frac{\log p}{n}} ) =O(p^{-M}). 
\end{equation}

We complete the proof by combining (\ref{supp:eq1})-(\ref{supp:eq4}) with a union bound argument.

\end{proof}

\subsection*{S2. Proof of Lemma 2}
\label{subsec:lemma2}

\begin{proof}
It is easy to check that $E(\check{H}_j^{(k)}) = 0$ and $\text{var}(\check{H}_j^{(k)}) = 1$. The marginal sub-Gaussian distribution assumption in assumption (C1) implies that $\check{H}_j^{(k)}$ has finite Orlicz $\psi_1$-norm (i.e., sub-exponential distribution with finite constants). Therefore, $(\check{H}_j^{(k)})^2 - 1$ is centered random variable with finite Orlicz $\psi_{1/2}$-norm. Note that $\check{H}_j^{(k)}$ are independent for $k \in [K]$. The result follows from equation (3.6) of \cite{adamczak2011restricted}.
\end{proof}

\subsection*{S3. Proof of Lemma 3}
\label{subsec:lemma3}

\begin{proof}
Note that $\check{H}_j^{(k)} - \hat{H}_j^{(k)} = \frac{\sqrt{n} \bar{X}^{(k)}_j  \sqrt{n}\bar{Y}^{(k)} }{\sqrt{n\theta^{(k)}_j}}$. By assumption (C1), we have that $E(\sqrt{n} \bar{X}^{(k)}_j )=E(\sqrt{n} \bar{Y}^{(k)} )=0$, $\text{var}(\sqrt{n} \bar{X}^{(k)}_j)=\text{var}(\sqrt{n} \bar{Y}^{(k)})=1$, and both $\sqrt{n} \bar{X}^{(k)}_j$ and $\sqrt{n} \bar{Y}^{(k)}$ are sub-Gaussian with bounded constants. Therefore, the first equation follows from Bernstein inequality (e.g., Definition 5.13 in \cite{vershynin2010introduction}) applied to centered sub-exponential variable $\sqrt{n}\bar{X}^{(k)}_j\cdot\sqrt{n}\bar{Y}^{(k)}$, noting $\theta^{(k)}_j \ge \tau_0$ by assumption (C1). The second equation follows from the first one, $\log^3 p=o(n)$, and a Bernstein inequality (e.g., Corollary 5.17 in \cite{vershynin2010introduction}) applied to the sum of centered sub-exponential variables $\check{H}_j^{(k)}$. 
\end{proof}

\bibliographystyle{apalike}
\bibliography{TSA}

\end{document}